%% file: s5ALC-correspondence-full.tex
\documentclass{article}
\usepackage{myijcai17}
\usepackage{stfloats}

\usepackage{times}
\usepackage{microtype}

\usepackage{amsthm}
\usepackage{amsmath, amsfonts, amssymb, stmaryrd, bbm, bussproofs,xspace}
\usepackage[matrix,arrow,curve]{xy}
\usepackage{textcomp}
\usepackage{url}
\usepackage{hyperref}
\usepackage[author=anonymous,nomargin,marginclue,footnote,status=draft]{fixme}
\FXRegisterAuthor{ls}{als}{LS}
\FXRegisterAuthor{pw}{apw}{PW}

\newcommand{\domain}{\mathit{domain}}
\newcommand{\world}{\mathit{world}}

\newcommand{\ef}{\cong}
\newcommand{\bisim}{\approx}
\newcommand{\ST}{\mathsf{ST}}
\newcommand{\shortiff}{\Leftrightarrow}

\newcommand{\Land}{\bigwedge}

\newcommand{\CI}{\mathcal{I}}
\newcommand{\CJ}{\mathcal{J}}

\renewcommand{\theta}{\vartheta}

\newcommand{\CN}{\mathsf{N}_{\mathsf{C}}}

\newcommand{\RN}{{\mathsf{N}_{\mathsf{R}}}}
\newcommand{\ALC}{\mathcal{ALC}}
\newcommand{\EL}{\mathcal{EL}}

\makeatletter
\newcommand*{\@old@slash}{}\let\@old@slash\slash
\def\slash{\relax\ifmmode\delimiter"502F30E\mathopen{}\else\@old@slash\fi}
\makeatother

\newlength{\myboxwidth}
	{\end{center}\end{figure*}}

\newcounter{blubber}

\newenvironment{myenumerate}
{\begin{enumerate}
\setlength{\itemsep}{0pt}
    \setlength{\leftmargin}{0pt}
    \setlength{\itemindent}{0pt}
}
{\end{enumerate}}

\newenvironment{myitemize}
{\begin{itemize}
\setlength{\itemsep}{0.1ex}
\setlength{\parsep}{0cm}
}
{\end{itemize}}

\title{A Characterization Theorem for a Modal Description Logic}
\author{Paul Wild \and Lutz Schr\"oder\thanks{Work forms part of DFG project \emph{ProbDL2} (SCHR 1118/6-2)}\\
Friedrich-Alexander-Universit\"{a}t Erlangen-N\"{u}rnberg}

\input{defsplain}

\newcommand{\SfALC}{S5_{\ALC}}
\newcommand{\SfALCloc}{S5^{\mathit{loc}}_{\ALC}}
\newcommand{\SfFOL}{S5_{\mathit{FOL}}}

\newcommand{\nf}[2]{C^{#1}_{#2}}
\newcommand{\abisim}{\approx^{\mathsf{alt}}}
\newcommand{\rk}{\textnormal{rk}}

\begin{document}

\maketitle

\begin{abstract} Modal description logics feature modalities that capture dependence of knowledge on parameters such as time, place, or the information state of agents. E.g., the logic $\SfALC$ combines the standard description logic $\ALC$ with an $S5$-modality that can be understood as an epistemic operator or as representing (undirected) change. This logic embeds into a corresponding modal first-order logic $\SfFOL$. We prove a modal characterization theorem for this embedding, in analogy to results by van Benthem and Rosen relating $\ALC$ to standard first-order logic: We show that $\SfALC$ with only local roles is, both over finite and over unrestricted models, precisely the bisimulation invariant fragment of $\SfFOL$, thus giving an exact description of the expressive power of $\SfALC$ with only local roles.
\end{abstract}

\section{Introduction}

Modal description logics extend the static knowledge model of standard
description logics by adding modalities capturing, e.g., the temporal
evolution of the state of the world or the dependence of knowledge on
the information available to individual agents. Their semantics is
typically \emph{two-dimensional}~\cite{GabbayEA02}, i.e.\ it is
defined over interpretations involving two sets of \emph{individuals}
and \emph{worlds}, respectively, and concepts are interpreted as
subsets of the Cartesian product of these two sets. For instance,
temporal description logics (surveyed, e.g., by Lutz et
al.~\shortcite{LutzEA08}) have a frame structure on the set of worlds, in
the same way as in the semantics of standard temporal logics such as
CTL; they support statements such as `every person that is currently a
child will eventually become an adult in the future'.

A simpler variant of the same idea is to give up directedness of
temporal evolution and instead introduce a modality that reads `at
some other point in time', so that, continuing the previous example,
one can express only that that every person that is currently a child
is an adult at some other time. This coarser granularity buys a
simplification of the semantics in which the set of worlds is just a
set (equivalently, a frame whose transition relation is an
equivalence), i.e.\ a model of the modal logic~$S5$. Modal description
logics with an $S5$ modality have been used prominently as
\emph{description logics of change}, and are able to encode a
restricted form of temporal entity-relationship models if the
description logic is strong enough (specifically, contains
$\mathcal{ALCQI}$)~\cite{ArtaleEA07}.

One of the simplest description logics of change in this sense is
$\SfALC$, i.e.\ the extension of the standard description logic
$\ALC$~\cite{BaaderEA03} with an $S5$ change modality. In fact, there
are many other readings for the $S5$ modality. In particular, $S5$
modalities standardly feature in epistemic logics, and indeed $\SfALC$
was originally introduced as an epistemic description
logic~\cite{WolterZakharyaschev99}. As a variant of this view,
$\SfALC$ and its $\EL$ fragment have been considered as a corner case
of probabilistic description logics for subjective uncertainty, with
probabilities mentioned in concepts restricted to $0$
or~$1$~\cite{Gutierrez-BasultoEA17}. In the current work, we focus on
$\SfALC$ as one of the most basic modal description logics, and use it
as a starting point for the \emph{correspondence theory} of modal
description logics.

Specifically, $\SfALC$ embeds as a fragment into the modal first-order
logic $\SfFOL$, which extends standard first-order logic with an $S5$
modality and lives over the same type of semantic structures as
$\SfALC$. This situation is analogous to the one with $\ALC$ itself,
which embeds as a fragment into ordinary first-order logic (FOL). For
$\ALC$, it is straightforward to check that its concepts are
\emph{bisimulation invariant}, i.e.\ bisimilar individuals satisfy the
same $\ALC$-concepts. This constitutes in effect an upper bound on the
expressivity of $\ALC$: any property that fails to be bisimulation
invariant (such as `individual $x$ is related to itself under role
$r$') is not expressible in $\ALC$. Remarkably, it can be shown that
this is also a lower bound: \emph{every} bisimulation invariant
\emph{first-order} property can be expressed in $\ALC$, a fact first
proved by van Benthem~\shortcite{BenthemThesis} and later shown to
hold true also over finite structures by Rosen~\shortcite{Rosen97}. In
other words, $\ALC$ is precisely the bisimulation invariant fragment
of FOL; we refer to theorems of this type as \emph{modal
  characterization theorems}.  In this terminology, the object of this
paper is to establish a modal characterization theorem for
$\SfALCloc$, the fragment of $\SfALC$ determined by admitting only
local (i.e.\ non-modalized) roles: We show that both over unrestricted
and over finite interpretations, \emph{$\SfALCloc$ is precisely the
  bisimulation invariant fragment of $\SfFOL$}, where both
bisimulation invariance and equivalence to a modal formula are
understood over two-dimensional interpretations. Technically, we
follow a generic recipe suggested by Otto~\shortcite{Otto04}, which
relies on \emph{locality} w.r.t.\ Gaifman distance. In fact, the main
challenge in our proof is to identify a suitable notion of Gaifman
distance for $\SfFOL$, and relate it to numbers of rounds played in
bisimulation games (Remark~\ref{rem:challenge}). Summing up, we pin
down the exact expressivity of $\SfALCloc$ as a fragment of $\SfFOL$;
to our best knowledge, this is the first time a modal characterization
theorem is obtained for a many-dimensional modal logic or a modal
description logic.\medskip

\noindent This paper is a full version of a shorter conference
paper~\cite{WildSchroder17-ijcai}. In the main text, proofs are sometimes
omitted or only sketched; full proofs are in the appendix.

\paragraph{Related Work} In the one-dimensional case, the original van
Benthem / Rosen characterization theorem has been extended in various
directions, e.g.\ for logics with frame conditions~\cite{DawarOtto05},
coalgebraic modal logics~\cite{SchroderEA15}, fragments of
XPath~\cite{tenCateEA10,FigueiraEA15,AbriolaEA17}, neighbourhood
logic~\cite{HansenEA09}, modal and first order logic with team
semantics~\cite{KontinenEA15}, modal $\mu$-calculi (within monadic
second order logics)~\cite{JaninWalukiewicz95,EnqvistEA15}, and PDL
(within weak chain logic)~\cite{Carreiro15}.

In the many-dimensional setting, all existing characterization results
that we are aware of look in the other direction, from the perspective
of modal first-order logics: they characterize modal first-order
logics as fragments of even more expressive two-sorted first order
logics that make the worlds explicit. Specifically, van
Benthem~\shortcite{Benthem01} proves this for unrestricted frames in
the modal dimension, i.e., in the nomenclature scheme we use here, for
$K_{\mathit{FOL}}$, while Sturm and Wolter~\shortcite{SturmWolter01}
characterize $\SfFOL$ as as a fragment of two-sorted FOL; in both
cases, the relevant notion of equivalence is essentially bisimilarity
in the modal dimension, and Ehrenfeucht-Fra\"iss\'e equivalence in the
first-order dimension. Both results are proved only over unrestricted
models, and their proofs rely on compactness. The characterization
theorem for $\SfFOL$ can in fact be combined with our characterization
of $\SfALC$ as a fragment of $\SfFOL$ to obtain, over unrestricted
models, a stronger characterization of $\SfALC$ as the
bisimulation invariant fragment of the two-sorted first-order
correspondence language (Corollary~\ref{cor:benthem} below).

\section{\!$S5$-modalized $\ALC$ and FOL}
\label{sec:s5alc}

\noindent We recall the syntax of modalized $\ALC$ as introduced by
Wolter and Zakharyaschev~\shortcite{WolterZakharyaschev99},
restricting to a single modality: Concepts $C,D$ of $\SfALCloc$
($S5$-\emph{modalized $\ALC$ with only local roles}) are given by the
grammar
\begin{equation*}
  C,D::= A \mid \neg C \mid C\sqcap D\mid \exists r.\,C\mid \Box C
\end{equation*}
where as usual $A$ ranges over a set $\CN$ of \emph{(atomic) concept
  names} and $r$ over a set $\RN$ of \emph{role names}. The remaining
Boolean connectives $\top,\bot,\sqcup$, as well as universal
restrictions $\forall r.\,C$, are encoded as usual. The \emph{rank} of
an $\SfALCloc$-concept~$C$ is the maximal nesting depth of modalities
$\Box$ and existential restrictions $\exists r$ in $C$ (e.g.\
$\exists r.\,\Box A$ has rank~$2$).

An \emph{$S5$-interpretation}
\begin{equation*}
\CI=(W^\CI,\Delta^\CI,((-)^{\CI,w})_{w\in W^\CI})
\end{equation*}
consists of nonempty sets $W^\CI$, $\Delta^\CI$ of of \emph{worlds}
and \emph{individuals}, respectively, and for each world $w\in W^\CI$
a standard $\ALC$ interpretation $(-)^{\CI,w}$ over $\Delta^\CI$,
i.e.\ for each concept name $A\in\CN$ a subset
$A^{\CI,w}\subseteq\Delta^\CI$, and for each role name $r\in\RN$ a
binary relation $r^{\CI,w}\subseteq\Delta^{\CI}\times\Delta^{\CI}$. We
refer to $\Delta^\CI$ as the \emph{domain} of $\CI$. The
interpretation $C^{\CI,w}\subseteq\Delta^\CI$ of a composite concept
$C$ at a world $w$ is then defined recursively by the usual clauses
for the $\ALC$ constructs
($(\neg C)^{\CI,w}=\Delta^\CI\setminus C^{\CI,w}$,
$(C\sqcap D)^{\CI,w}=C^{\CI,w}\cap D^{\CI,w}$,
$(\exists r.\,C)^{\CI,w}=\{d\in\Delta^\CI\mid\exists e\in
C^{\CI,w}.\,(d,e)\in r^{\CI,w}\}$), and by
\begin{equation*}
  (\Box C)^{\CI,w}=\{d\in\Delta^\CI\mid\forall v\in W^\CI.\,d\in C^{\CI,v}\}.
\end{equation*}
In words, $\Box C$ denotes the set of individuals that belong to~$C$
in all worlds. As usual, we write $\Diamond$ for the dual of $\Box$,
i.e.\ $\Diamond C$ abbreviates $\neg\Box\neg C$ and denotes the set of
all individuals that belong to $C$ in some world. We write
$\CI,w,d\models C$ if $d\in C^{\CI,w}$.

$S5$-interpretations are \emph{two-dimensional} in the sense that
concepts are effectively interpreted as subsets of Cartesian products
$W^\CI\times\Delta^\CI$, and the modalities $\Box$ and $\exists r$ are
interpreted by relations that move only in one dimension of the
product: $\Box$ moves only in the world dimension and keeps the
individual fixed, and vice versa for~$\exists r$. 
Thus, $\SfALCloc$ and $\SfALC$
(Remark~\ref{rem:s5alc}) are examples of \emph{many-dimensional modal
  logics}~\cite{MarxVenema97,GabbayEA02}. 

As indicated in the introduction, there are various readings that can
be attached to the modality $\Box$. E.g.\ if we see $\Box$ as a
\emph{change modality}~\cite{ArtaleEA07}, and, for variety, consider
spatial rather than temporal change, then the concept
\begin{gather*}
  \exists\,\mathsf{isMarriedTo}.\,
  (\neg C\sqcap \Diamond C)\\
  \text{where } C = \exists\,\mathsf{wantedBy}.\,\mathsf{LawEnforcement}
\end{gather*}
describes persons married to fugitives from the law, i.e.\ to persons
that are wanted by the police in some place but not here. As an
example where we read $\Box$ as an epistemic modality `I know
that'~\cite{WolterZakharyaschev99,GabbayEA02}, the concept
\begin{equation*}
  \Box\exists\,\mathsf{has}.\,(\mathsf{Gun}\sqcap\mathsf{Concealed}\sqcap\Diamond\mathsf{Loaded})
\end{equation*}
applies to people who I know are armed with a concealed gun that as
far as I know might be loaded.

A more expressive modal language is \emph{$S5$-modalized first-order
  logic with constant domains}, which we briefly refer to as
$\SfFOL$. Formulas $\phi,\psi$ of $\SfFOL$ are given by the grammar
\begin{equation*}
  \phi,\psi::= R(x_1,\dots,x_n)\mid x = y\mid \neg\phi\mid\phi\sqcap\psi\mid\exists x.\,\phi\mid\Box\phi
\end{equation*}
where $x,y$ and the $x_i$ are variables from a fixed countably
infinite reservoir and $R$ is an $n$-ary predicate from an underlying
\emph{language} of predicate symbols with given arities. The
quantifier $\exists x$ binds the variable $x$, and we have the usual
notions of free and bound variables in formulas. The \emph{rank} of a
formula $\phi$ is the maximal nesting depth of modalities $\Box$ and
quantifiers $\exists x$ in~$\phi$; e.g.\
$\exists x.\,\Box(\forall y.\,r(x,y))$ has rank~$3$. This is exactly
the $S5$ modal first-order logic called $\mathcal{QML}$ by Sturm and
Wolter~\shortcite{SturmWolter01}. From now on we fix the language to
be the \emph{correspondence language} of $\SfALC$, which has a unary
predicate symbol $A$ for each concept name~$A$ and a binary predicate
symbol $r$ for each role name $r$.  The semantics of $\SfFOL$ is then
defined over $S5$-interpretations, like $\SfALCloc$. It is given in
terms of a satisfaction relation~$\models$ that relates an
interpretation $\CI$, a world $w\in W^\CI$, and a \emph{valuation}
$\eta$ assigning a value $\eta(x)\in\Delta^\CI$ to every variable $x$
on the one hand to a formula $\phi$ on the other hand.  The
relation~$\models$ is defined by the expected clauses for Boolean
connectives, and
\begin{gather*}
  \CI,w,\eta\models R(x_1,\dots,x_n)  \shortiff (\eta(x_1),\dots,\eta(x_n))\in R^{\CI,w}\\
  \CI,w,\eta\models x=y   \shortiff \eta(x)=\eta(y)\\
  \CI,w,\eta\models \exists x.\,\phi  \shortiff \CI,w,\eta[x\mapsto d]\models\phi
                                       \text{ for some $d\in\Delta^\CI$}\\
  \CI,w,\eta\models\Box\phi  \shortiff \CI,v,\eta\models\phi
                              \text{ for all $v\in W^\CI$}
\end{gather*}
(where $\eta[x\mapsto d]$ denotes the valuation that maps $x$ to $d$
and otherwise behaves like $\eta$). That is, the semantics of the
first-order constructs is as usual, and that of $\Box$ is as in
$\SfALC$. We often write valuations as vectors
$\bar d=(d_1,\dots,d_n)\in(\Delta^\CI)^n$, which list the values assigned to
variables $x_1,\dots,x_n$ if the free variables of $\phi$ are
contained in $\{x_1,\dots,x_n\}$.

To formalize the obvious fact that \emph{$\SfALCloc$ is a fragment of
  $\SfFOL$}, we extend the usual standard translation to $\SfALC$: a
translation $\ST_x$ that maps $\SfALC$-concepts $C$ to
$\SfFOL$-formulas $\ST_x(C)$ with a single free variable $x$ is given
by
\begin{align*} 
  \ST_x(A) & = A(x)\\
  \ST_x(\exists r.\,C) & =\exists y.\,(r(x,y)\sqcap\ST_y(C))\quad\text{($y$ fresh)}
\end{align*}
and commutation with all other constructs. Then $\ST_x$ preserves the
semantics, i.e.\
\begin{lem}
  For every $\SfALCloc$-concept $C$, interpretation $\CI$, $w\in W^\CI$,
  and $d\in\Delta^\CI$, we have
  \begin{equation*}
    \CI,w,d\models C\quad\text{iff}\quad \CI,w,d\models \ST_x(C).
  \end{equation*}
\end{lem}

\begin{rem}\label{rem:s5alc}
  Modalized $\ALC$ has been extended with \emph{modalized
    roles}~\cite{WolterZakharyaschev99b}, i.e.\ roles of the form
  $\Box r$ or $\Diamond r$, interpreted as
  \begin{align*}
    (\Box r)^{\CI,w}&=\{(d,e)\mid \forall w\in W^\CI.\,(d,e)\in r^{\CI,w}\}\\
    (\Diamond r)^{\CI,w}&=\{(d,e)\mid \exists w\in W^\CI.\,(d,e)\in r^{\CI,w}\}.
  \end{align*}
  The $S5$-modalized description logic in this extended sense has been
  termed $\SfALC$ by Gabbay et al.~\shortcite{GabbayEA02}; so in our
  notation $\SfALCloc$ is the fragment of $\SfALC$ without modalized
  roles. Since modalized roles $\Box r$, $\Diamond r$ have an
  interpretation that is independent of the world while that of basic
  roles $r$ varies between worlds, the latter are called \emph{local}
  roles, explaining the slightly verbose terminology used above. We
  will see that $\SfALC$ fails to be bisimulation invariant, and is
  therefore strictly more expressive than $\SfALCloc$.
\end{rem}

\section{Bisimulation and Invariance}

\noindent We proceed to introduce the relevant notion of bisimulation
for $S5$-interpretations. This is just the usual notion of
bisimilarity, specialized to the two-dimensional shape of
$S5$-interpretations and the $S5$ structure of the world
dimension; explicitly:
\begin{defn}[Bisimulation]
  A \emph{bisimulation} between interpretations $\CI$, $\CJ$ is a relation
  \begin{equation*}
    R \subseteq (W^\CI\times\Delta^\CI)\times(W^\CJ\times\Delta^\CJ)
  \end{equation*}
  such that whenever $(w,d) R (v,e)$, then
  \begin{enumerate}
  \item $d \in A^{\CI,w}$ iff $e \in A^{\CJ,v}$ for all $A \in \CN$;
  \item for every $w'\in W^\CI$ there is $v'\in W^\CJ$ such that
    $(w',d)R(v',e)$;
  \item Same with the roles of $\CI$ and $\CJ$ interchanged.
  \item for every $(d,d') \in r^{\CI,w}$ ($r\in\RN$) there is $e'$
    such that $(e,e') \in r^{\CJ,v}$ and $(w, d') R (v,e')$
  \item Same with the roles of $\CI$ and $\CJ$ interchanged.
  \end{enumerate}
  We say that $\CI,w,d$ and $\CJ,v,e$ are \emph{bisimilar}, and
  write
  \begin{equation*}
    \CI,w,d\bisim \CJ,v,e
  \end{equation*}
  if there exists a bisimulation $R$ such that $(w,d)R(v,e)$.
\end{defn}
\noindent We record explicitly that $\SfALCloc$ is \emph{bisimulation
  invariant}, a fact that is immediate from bisimulation invariance
of basic multi-modal logic (over all interpretations, including
$S5$-interpretations). As a general manner of speaking, whenever~$P$
is any property that applies to triples $\CI,w,d$ consisting of an
$S5$-interpretation $\CI$, $w\in W^\CI$, and $d\in\Delta^\CI$ (e.g.\
$P$ could be an $\SfALCloc$-concept or an $\SfFOL$-formula with one
free variable), then we say that~$P$ is \emph{bisimulation invariant},
or just \emph{$\bisim$-invariant}, if whenever $\CI,w,d\bisim\CJ,v,e$
then $\CI,w,d$ has property $P$ iff $\CJ,v,e$ has property $P$.  We
will extend this terminology without further comment to other notions
of equivalence that we introduce later, such as bisimilarity up to
finite depth and Ehrenfeucht-Fra\"iss\'e equivalence. Moreover, we
will consider restrictions of these notions to finite
$S5$-interpretations; e.g.\ \emph{bisimulation-invariance over finite
  $S5$-interpretations} of a property~$P$ is defined like
bisimulation-invariance of~$P$ above but with $\CI$ and $\CJ$ assumed
to be finite. In these terms, we have
\begin{lem}[Bisimulation invariance]\label{lem:bisim-invar}
  Every $\SfALCloc$-concept is $\bisim$-invariant.
\end{lem}
\begin{expl}\label{expl:modalized-roles}
  As indicated in the introduction, bisimulation invariance is an
  upper bound on the expressivity of $\SfALC$. As an extremely simple
  example, the formula $r(x,x)$ of $\SfFOL$ fails to be
  $\bisim$-invariant and is therefore, by Lemma~\ref{lem:bisim-invar},
  not equivalent to (the standard translation of) any
  $\SfALCloc$-concept. Bisimulation invariance also separates
  $\SfALCloc$ from $\SfALC$ (Remark~\ref{rem:s5alc}): the
  $\SfALC$-concept $\exists\Diamond r.\, A$ fails to be
  $\bisim$-invariant and is therefore not expressible in $\SfALCloc$.
\end{expl}


\paragraph*{Bisimulation games} As usual, bisimilarity can
equivalently be captured in terms of games. Explicitly:
\begin{defn}[Bisimulation game]\label{def:bisim-game}
  Let $\CI$, $\CJ$ be $S5$-interpretations, and let
  $(w_0,d_0)\in W^\CI\times\Delta^\CI$,
  $(v_0,e_0)\in W^\CJ\times\Delta^\CJ$.  The \emph{bisimulation game for
    $\CI,w_0,d_0$ and $\CJ,v_0,e_0$} is played by players $S$
  (\emph{Spoiler}) and~$D$ (\emph{Duplicator}), where~$D$ means to
  establish bisimilarity and~$S$ aims to disprove it. A
  \emph{configuration} of the game is a quadruple
  $((w,d),(v,e))\in
  (W^\CI\times\Delta^\CI)\times(W^\CJ\times\Delta^\CJ)$,
  with $((w_0,d_0),(v_0,e_0))$ being the initial configuration. A
  \emph{round} consists of one move by $S$ and a subsequent move by
  $D$, with the following alternatives in the current configuration
  $((w,d),(v,e))$:
  \begin{enumerate}
  \item\label{item:move-w-i} $S$ may pick a world $w'\in W^\CI$, and
    $D$ then needs to pick a world $v'\in W^\CJ$; the new
    configuration then is $((w',d),(v',e))$.
  \item \label{item:move-w-j} Same with the roles of $\CI$ and $\CJ$
    interchanged.
  \item \label{item:move-delta-i} $S$ may pick a role $r\in\RN$ and an
    individual $d'\in\Delta^\CI$ such that $(d,d')\in r^{\CI,w}$. Then
    $D$ needs to pick an individual $e'\in\Delta^\CJ$ such that
    $(e,e')\in r^{\CJ,v}$; the new configuration reached is
    $((w,d'),(v,e'))$.
  \item \label{item:move-delta-j} Same with the roles of $\CI$ and
    $\CJ$ interchanged.
  \end{enumerate}
  We will call the first two kinds of moves \emph{$W$-moves} and the
  other two kinds \emph{$\Delta$-moves}.  If one of the players cannot
  move, then the other one wins. A configuration $((w,d),(v,e))$ is
  \emph{winning for $S$} if $d\in A^{\CI,w}$ and $e\notin A^{\CJ,v}$
  for some concept name $A\in\CN$, or vice versa; and $S$ wins if a
  winning configuration for $S$ is reached. Infinite plays that do no
  visit a winning configuration for $S$ are won by $D$.
\end{defn}
\noindent The following is then standard:
\begin{lem}
  We have $\CI,w,d\bisim \CJ,v,e$ iff $D$ wins the bisimulation game
  for $\CI,w,d$ and $\CJ,v,e$.
\end{lem}

\noindent The bisimulation game can be restricted to a finite number
of rounds, capturing bisimulation up to finite depth:

\begin{defn}[Finite-depth bisimulation]
  The \emph{$n$-round bisimulation game} for $n\ge 0$ is played in the
  same way as the bisimulation game but only for at most $n$
  rounds. The winning conditions are the same as in the bisimulation
  game except that $D$ now wins if no winning configuration for $S$
  has been reached after $n$ rounds. We say that $\CI,w,d$ and
  $\CJ,v,e$ are \emph{depth-$n$ bisimilar}, and write
  \begin{equation*}
    \CI,w,d\bisim_n\CJ,v,e,
  \end{equation*}
  if $D$ wins the $n$-round bisimulation game for $\CI,w,d$ and
  $\CJ,v,e$.
\end{defn}
\noindent Again, the following is standard:
\begin{lem}[Invariance under finite-depth bisimulation]
  Every $\SfALCloc$-concept $C$ of rank at most $n$ is
  $\bisim_n$-invariant.
\end{lem}

\noindent For technical purposes, we shall need a normalization of the
bisimulation game based on the observation that due to the $S5$
structure of the worlds, $S$ can never gain an advantage from playing
more than one consecutive $W$-move. Formally:

\begin{defn}[Alternating bisimulation game]\label{def:alt-bisim-game}
  The \emph{alternating bisimulation game} is played like the
  bisimulation game but with a restriction on the sequence of moves:
  Each \emph{round} in the alternating bisimulation game consists of
  two \emph{phases},
  \begin{enumerate}
  \item\label{phase:w} $S$ may decide to make a $W$-move, and in this
    case~$D$ also makes a $W$-move, according to
    Item~\ref{item:move-w-i} or Item~\ref{item:move-w-j} of
    Definition~\ref{def:bisim-game}, and then
  \item\label{phase:delta} $S$ and $D$ each play exactly one
    $\Delta$-move according to Item~\ref{item:move-delta-i} or
    Item~\ref{item:move-delta-j} of Definition~\ref{def:bisim-game}
  \end{enumerate}
  (where $D$ needs to avoid winning configurations for~$S$ at all
  times). The alternating bisimulation game also comes in two variants, the
  unbounded and the $n$-round game, with the proviso that at the end
  of an $n$-round game, there may be one extra phase of
  type~\ref{phase:w} above. We write
  \begin{equation*}
    \CI,w,d\abisim\CJ,v,e\quad\text{and}\quad
    \CI,w,d\abisim_n\CJ,v,e
  \end{equation*}
  if~$D$ has a winning strategy in the alternating and in the
  $n$-round alternating bisimulation game for $\CI,w,d$ and $\CJ,v,e$,
  respectively.
\end{defn}
\noindent The unrestricted game is equivalent to the alternating game
in the following sense:
\begin{lem}\label{lem:abisim}
  For interpretations $\CI,\CJ$ and
  $(w,d) \in W^\CI \times \Delta^\CI$,
  $(v,e) \in W^\CJ \times \Delta^\CJ$:
\begin{enumerate}
\item If $\CI,w,d \bisim_{2n+1} \CJ,v,e$, then $\CI,w,d \abisim_n \CJ,v,e$.
\item If $\CI,w,d \abisim_n \CJ,v,e$, then $\CI,w,d \bisim_n \CJ,v,e$.
\item $\CI,w,d \bisim \CJ,v,e$ iff $\CI,w,d \abisim \CJ,v,e$.
\end{enumerate}
\end{lem}

\section{The Modal Characterization Theorem}\label{sec:main}

\noindent We proceed to state our main result and sketch its proof:
\emph{$\SfALCloc$ is the bisimulation-invariant fragment of $\SfFOL$},
both over finite and over unrestricted $S5$-interpretations. Formally,

\begin{thm}[Modal characterization]\label{thm:rosen}
  Let $\phi = \phi(x)$ be an $\SfFOL$-formula with one free variable
  $x$. If $\phi$ is $\bisim$-invariant (over finite
  $S5$-interpretations), then there exists an $\SfALCloc$-concept $C$
  such that $\phi$ is logically equivalent to $\ST_x(C)$ (over finite
  $S5$-interpretations). Moreover, the rank of $C$ is
  exponentially bounded in the rank of $\phi$. 
\end{thm}
\noindent While modal characterization theorems over unrestricted
structures can often be proved using model-theoretic tools such as
compactness~\cite{BenthemThesis}, proofs that apply also to finite
structures typically need to work with some form of
\emph{locality}~\cite{Otto04}. In the basic, one-dimensional case,
this is Gaifman locality~\cite{Gaifman82}, which is based on the
notion of \emph{Gaifman distance} in a first-order model: The
\emph{Gaifman graph} of the model connects two of its points if they
occur together in some tuple that is in the interpretation of some
relation in the model, and the Gaifman distance is then just the graph
distance in the Gaifman graph. We adapt these notions for our purposes
as follows.


\begin{defn}
  The \emph{Gaifman graph} of an $S5$-interpretation~$\CI$ is the
  undirected graph with vertex set $\Delta^\CI$ that has an edge
  between $d$ and $e$ iff $d\neq e$ and either $(d,e)\in r^{\CI,w}$ or
  $(e,d)\in r^{\CI,w}$ for some role name~$r$ and some $w \in W^\CI$.
  The \emph{Gaifman distance}
  $D: \Delta^\CI \times \Delta^\CI \rightarrow \mathbb{N} \cup
  \{\infty\}$
  is just graph distance (length of the shortest connecting path) in
  the Gaifman graph, and for any tuple
  $\bar d = (d_1,\dots,d_k) \in (\Delta^\CI)^k$, the
  \emph{neighbourhood} $U^\ell(\bar d)$ of $\bar d$ with radius $\ell$
  is given by
  \[\textstyle
  U^\ell(\bar d) = \{e \in \Delta^\CI \mid \min_{i=1}^k D(d_i,e) \le
  \ell \}.
  \]
\end{defn}
\begin{rem}\label{rem:challenge}
  It may be slightly surprising that Gaifman graphs for
  $S5$-interpretations live only in the individual dimension, so that
  implicit steps between worlds are effectively discounted (a point
  where the $S5$ structure on worlds becomes important). The technical
  reason for this is that it does not seem easily possible to include
  the worlds in the Gaifman graph without creating unduly short
  paths. The fact that world steps count $0$ in the Gaifman distance
  creates a certain amount of tension with the fact that bisimulation
  games do feature explicit $W$-moves
  (Definition~\ref{def:bisim-game}). Our alternating bisimulation
  games (Definition~\ref{def:alt-bisim-game}) serve mainly to address
  this point.
\end{rem}
\begin{defn}[Locality]
  The \emph{restriction} $\CI|_U$ of an $S5$-interpretation $\CI$ to a
  subset $U\subseteq\Delta^\CI$ is given by $W^{\CI|_U}=W^\CI$,
  $\Delta^{\CI|_U}=U$, $A^{\CI|_U,w}=A^{\CI,w}\cap U$ for $A\in\CN$,
  and $r^{\CI|_U,w}=r^{\CI,w}\cap(U\times U)$ for $r\in\RN$. An
  $\SfFOL$-formula $\phi$ with $k$ free variables is
  \emph{$\ell$-local} for $\ell\ge 0$ if for every $S5$-interpretation
  $\CI$, $w\in W^\CI$, and and $\bar d\in(\Delta^\CI)^k$,
  \begin{equation*}
    \CI,w,\bar d\models\phi\quad\text{iff}\quad
    \CI|_{U^\ell(\bar d)},w,\bar d\models\phi.
  \end{equation*}
\end{defn} 
\noindent In these terms, we organize the proof of our main result as
follows, following a generic strategy proposed by
Otto~\shortcite{Otto04}:
\begin{proof}[Proof of Theorem~\ref{thm:rosen} (Sketch)]
  For $\phi$ $\bisim$-invariant of rank~$n$, we prove the
  following steps in order:
  \begin{itemize}
  \item $\phi$ is $\ell$-local, where $\ell = 3^n$
  (Lemma~\ref{lem:mc1}).
  \item $\phi$ is $\bisim_{2\ell+1}$-invariant (Lemma~\ref{lem:mc2}).
  \item $\phi$ is equivalent to a concept of rank $2\ell+1$
  (Lemma~\ref{lem:mc3}).
  \end{itemize}
  (The locality bound is slightly generous, for simplicity.)
\end{proof}
\noindent Standard FOL comes with its own notion of invariance, with
respect to \emph{Ehrenfeucht-Fra\"iss\'e
  equivalence}~\cite{Libkin04}. This notion has been extended to
$\SfFOL$ by complementing it with bisimilarity in the world
dimension~\cite{Benthem01,SturmWolter01}. Here, we introduce a bounded
version of this equivalence, which we phrase in game-theoretic terms;
this will be instrumental in the proof of locality:

\begin{defn}[Bounded Ehrenfeucht-Fra\"iss\'e game for $\SfFOL$]\label{defn:ef}
  Let $\CI$, $\CJ$ be $S5$-interpretations, let
  $(w_0,d_0)\in W^\CI\times\Delta^\CI$,
  $(v_0,e_0)\in W^\CJ\times\Delta^\CJ$, and let $n \ge 0$. 
  The \emph{$n$-round Ehrenfeucht-Fra\"iss\'e game for $\CI,w_0,d_0$ and
  $\CJ,v_0,e_0$} is played by players $S$ and $D$.
  The \emph{configurations} are quadruples $((w,\bar d),(v,\bar e))$,
  where $w \in W^\CI$, $v \in W^\CJ$ and $\bar d$ and $\bar e$ are
  finite sequences over $\Delta^\CI$ and $\Delta^\CJ$, respectively.
  The initial configuration is $((w_0,d_0),(v_0,e_0))$.
  The possible moves from configuration $((w,\bar d),(v,\bar e))$ are:
  \begin{enumerate}
  \item $S$ may pick a world $w'\in W^\CI$, and $D$ then needs to pick
    a world $v'\in W^\CJ$; the new configuration is 
    $((w',\bar d),(v',\bar e))$.
  \item Same with the roles of $\CI$ and $\CJ$ interchanged.
  \item $S$ may pick some $d \in \Delta^\CI$ and $D$ then needs to
  pick $e \in \Delta^\CJ$. The new configuration is $((w,\bar d
  d),(v,\bar e e))$.
  \item Same with the roles of $\CI$ and $\CJ$ interchanged.
  \end{enumerate}
  The winning conditions are as in the $n$-round bisimulation game,
  except that a configuration is now \emph{winning for~$S$} if it
  fails to be a partial isomorphism. Here,
  $((w,(d_0,\dots,d_k)),(v,(e_0,\dots,e_k)))$ is a \emph{partial
    isomorphism} if
  \begin{itemize}
  \item for all $0 \le i,j \le k$, $d_i = d_j \Leftrightarrow e_i =
  e_j$; and
  \item for all $0 \le i_1,\dots,i_m\le k$ and $m$-ary relation
  symbols $R$, $(d_{i_1},\dots,d_{i_m})\in R^{\CI,w} \Leftrightarrow
  (e_{i_1},\dots,e_{i_m})\in R^{\CJ,v}$.
  \end{itemize}
  \noindent We say that $\CI,w_0,d_0$ and $\CJ,v_0,e_0$ are
  \emph{$S5$-Ehrenfeucht-Fra\"iss\'e equivalent up to depth $n$}, and write
  \begin{equation*}
    \CI,w_0,d_0 \ef_n \CJ,v_0,e_0,
  \end{equation*}
  if $D$ has a winning strategy in this game.
\end{defn}
\noindent As announced, $\SfFOL$ is invariant under
$S5$-Ehrenfeucht-Fra\"iss\'e equivalence. For the unbounded variant, this
has been shown in earlier work~\cite{SturmWolter01}; for our
bounded variant, invariance takes the following shape:
\begin{lem}[Bounded $S5$-Ehrenfeucht-Fra\"iss\'e
invariance]\label{lem:efequiv}
  Every $\SfFOL$-formula of rank at most $n$ with one free variable is
  $\ef_n$-invariant.
\end{lem}
\noindent We use this to prove locality:
\begin{lem}\label{lem:mc1}
  Let $\phi$ be a $\bisim$-invariant $\SfFOL$-formula of
  rank~$n$. Then $\phi$ is $\ell$-local for $\ell=3^n$.
\end{lem}
\begin{proof}[Proof (sketch)]
  Let $\CI$ be an $S5$-interpretation and
  $(w_0,d_0) \in W^\CI\times\Delta^\CI$.  Put
  $\CJ = \CI|_{U^\ell(d_0)}$; we need to show that
  $\CI,w_0,d_0 \models \phi \Leftrightarrow \CJ,w_0,d_0 \models \phi$.
  By $\bisim$-invariance, we can disjointly extend the domains of
  $\CI$ and $\CJ$ without affecting satisfaction of $\phi$. We thus
  extend both $\CI$ and $\CJ$ with $n$ copies of both $\CI$ and $\CJ$
  each, obtaining $\CI'$ and~$\CJ'$, respectively.

  By Lemma~\ref{lem:efequiv}, it suffices to show that
  $\CI',w_0,d_0\ef_n\CJ',w_0,d_0$. The winning strategy for $D$ is to
  maintain the following invariant, where we put $\ell_i = 3^{n-i}$
  for $0 \le i \le n$:
  
  \begin{quote}
    If $((w,\bar d),(v,\bar e))$ is the current configuration, with
    $\bar d = (d_0,\dots,d_i), \bar e = (e_0,\dots,e_i)$, then $w=v$
    and there is an isomorphism between $\CI'|_{U^{\ell_i}(\bar d)}$
    and $\CJ'|_{U^{\ell_i}(\bar e)}$ mapping each $d_j$ to $e_j$.
  \end{quote}
  $D$ maintains the invariant as follows: Whenever $S$ picks a new
  world in either interpretation, $D$ can just pick the same world in
  the other interpretation, as $\CI'$ and $\CJ'$ have the same set of
  worlds. Whenever~$S$ picks a new individual $d$ in
  $U^{2\ell_{i+1}}(\bar d)$ or $U^{2\ell_{i+1}}(\bar e)$ (where $\bar
  d = (d_0,\dots,d_i)$ and $\bar e = (e_0,\dots,e_i)$), then $d$ is in
  the domain or range of the isomorphism in the invariant, and $D$
  picks his response according to the isomorphism. Otherwise, $D$
  picks a `fresh' copy of the appropriate type ($\CI$ or $\CJ$,
  depending on where~$d$ lies) in the other interpretation and
  responds with $d$ in that copy.
\end{proof}

\noindent Having proved locality of $\bisim$-invariant formulas, we
next establish invariance even under finite-depth bisimilarity. To
this end, we need \emph{tree unravellings} of $S5$-interpretations:

\begin{defn}[Tree unravelling]
  Let $\CI$ be an interpretation and $d_0 \in \Delta^\CI$. The
  \emph{tree unravelling} $\CI^\ast_{d_0}$ of $\CI$ is the
  interpretation with set $W^{\CI^\ast_{d_0}}=W^\CI$ of worlds; with
  domain $\Delta^{\CI^\ast_{d_0}}$ consisting of all paths of the form
  $(d_0,\dots,d_k)$ such that for each $i\in\{0,\dots,k-1\}$,
  $(d_i,d_{i+1})\in r^{\CI,w}$ for some role name~$r$ and some world
  $w$; and with the following interpretations of concept and role
  names:
\begin{gather*}
  A^{\CI^\ast_{d_0},w} = \{ \bar d \in \Delta^{\CI^\ast_{d_0}} \mid
  \pi(\bar d) \in A^{\CI,w} \} \\
  r^{\CI^\ast_{d_0},w} = \{ (\bar d,\bar d d) \mid \bar d \in
  \Delta^{\CI^\ast_{d_0}}, (\pi(\bar d),d) \in r^{\CI,w} \}
\end{gather*}
where $\pi: (d_0,\dots,d_k) \mapsto d_k$ is projection to the last
entry.
\end{defn}

\noindent It is then easy to show that $\CI,w,d \bisim \CI^\ast_d,w,d$.
In fact, a bisimulation is given by the function $\pi$ (and identity
on the set of worlds). Also,
$\CI^\ast_d,w,d\abisim_\ell\CI^\ast_d|_{U^\ell(d)},w,d$.

\begin{lem}\label{lem:mc2}
  Let $\phi = \phi(x)$ be $\bisim$-invariant and
  $\ell$-local. Then~$\phi$ is $\bisim_{2\ell+1}$-invariant.
\end{lem}
\begin{proof}[Proof (sketch)]
  Let $\CI,w,d \bisim_{2\ell+1} \CJ,v,e$ and $\CI,w,d \models \phi$.
  We need to show that $\CJ,v,e \models \phi$. By
  Lemma~\ref{lem:abisim}, $\CI,w,d \abisim_\ell \CJ,v,e$. By
  $\bisim$-invariance of $\phi$, we may pass from $\CI$ and $\CJ$ to
  their unravellings, and by $\ell$-locality of $\phi$, we may then
  restrict those to the radius $\ell$ neighbourhoods of~$d$ and~$e$,
  respectively. The resulting interpretations
  $\CI^\ast_d|_{U^\ell(d)}$ and $\CJ^\ast_e|_{U^\ell(e)}$ then are
  trees of height at most $\ell$ in the individual dimension.
  
  Now
  $\CI^\ast_d|_{U^\ell(d)},w,d \abisim_\ell
  \CJ^\ast_e|_{U^\ell(e)},v,e$,
  i.e.\ $D$ wins the alternating $\ell$-round bisimulation game. Due
  to the tree structure on the domains, $D$'s winning strategy is also
  winning for the unbounded alternating bisimulation game, as
  eventually a leaf node will be reached and $S$ will not have a legal
  move in the second phase of a round. So, using
  Lemma~\ref{lem:abisim} again, $D$ wins the unbounded ordinary
  bisimulation game, and therefore $\CJ,v,e \models \phi$ by
  $\bisim$-invariance of $\phi$.
\end{proof}

\begin{rem}\label{rem:partunravel}
  In the case of finite interpretations (the `Rosen' part of the
  characterization theorem), there is a caveat: the tree unravelling
  of a finite interpretation is not finite in general, so we
  cannot use $\bisim$-invariance over finite interpretations to pass
  from interpretations to their unravellings. To remedy this, we work
  with \emph{partial unravellings} up to level $\ell$ instead. Such
  a partial unravelling is constructed by restricting the tree
  unravelling $\CI^\ast_{d_0}$ to the radius $\ell+1$ neighbourhood
  of $d_0$ and then identifying each leaf node $\bar d$ with the
  corresponding element $\pi(\bar d)$ in a fresh disjoint copy of
  $\CI$.  The resulting interpretation is clearly finite if~$\CI$ is
  finite, and readily shown to be bisimilar to $\CI$. Also,
  the radius $\ell$ neighbourhood of $d_0$ in the partial unravelling
  is a tree.
\end{rem}

\noindent Finally, we construct an equivalent $\SfALC$-concept for a
given formula that is invariant under finite-depth bisimulation. We
will make use of normal forms, as introduced by
Fine~\shortcite{fine1975normal}.

Since the formula $\phi$ is fixed, we can assume w.l.o.g.\ that $\CN$
and $\RN$ are finite sets $\CN=\{A_1,\dots,A_s\}$ and
$\RN=\{r_1,\dots,r_t\}$.

\begin{defn}\label{defn:nf}
  The sets $\mathsf{nf}_k$ and $\mathsf{at}_k$ of \emph{normal forms}
  and \emph{atoms} of rank $k \ge 0$, respectively, are defined by
  induction:
\begin{multline*}
  \mathsf{at}_k = \{ A_1,\dots,A_s \} \cup
  \{ \exists r_i.C \mid 1 \le i \le t, C \in \mathsf{nf}_{k-1} \}
  \cup \\
  \{ \Diamond C \mid C \in \mathsf{nf}_{k-1} \}
\end{multline*}
and $\mathsf{nf}_k$ is the set of finite conjunctions of the form
$\Land_{B\in\mathsf{at}_k} \varepsilon_B B$ (according to some fixed
total ordering on $\mathsf{at}_k$) where each $\varepsilon_B$ is
either nothing or negation. Moreover, $\mathsf{nf}_{-1} = \emptyset$
for convenience.
\end{defn}
\noindent 
These normal forms have the following properties:

\begin{myitemize}
\item For any $\CI,w,d$, there is exactly one normal form
  $C^k_{\CI,w,d}$ of rank $k$ such that $\CI,w,d \models
  \nf{k}{\CI,w,d}$.
\item We have $\CI,w,d \bisim_k \CJ,v,e$ iff
  $\nf{k}{\CI,w,d} = \nf{k}{\CJ,v,e}$.
\end{myitemize}
\begin{lem}\label{lem:mc3}
  Every $\bisim_k$-invariant $\SfFOL$-formula $\phi$ with one free
  variable~$x$ can be expressed as an $\SfALCloc$-concept of rank~$k$,
  namely
\[\textstyle
  \phi \equiv \ST_x\big(\bigvee_{\CI,w,d \models \phi} \nf{k}{\CI,w,d}\big).
\]
\end{lem}
\begin{proof}
  First, note that the above disjunction is finite, even though there
  may be infinitely many interpretations satisfying~$\phi$. We denote
  the arising $\SfALCloc$-concept by~$C$.

  For the implication from $\phi$ to $\ST_x(C)$, just note that if
  $\CI,w,d \models \phi$, then $\nf{k}{\CI,w,d}$ is one of the
  disjuncts in $C$. 

  For the reverse implication, let $\CI,w,d \models C$ and
  let~$\nf{k}{\CJ,v,e}$ be a disjunct in $C$ such that
  $\CI,w,d \models \nf{k}{\CJ,v,e}$. By the above properties of normal
  forms, it follows that $\nf{k}{\CI,w,d} = \nf{k}{\CJ,v,e}$ and
  therefore $\CI,w,d \bisim_k \CJ,v,e$. By definition,
  $\CJ,v,e \models \phi$, so $\CI,w,d \models \phi$ by
  $\bisim_k$-invariance of $\phi$, as desired.
\end{proof}
\noindent This completes the proof of Theorem~\ref{thm:rosen} as outlined
above.

\paragraph{Characterization within two-sorted FOL} 
The natural first-order correspondence language for
$\SfFOL$~\cite{SturmWolter01} is a two-sorted language with sorts
$\domain$ and $\world$; for every $n$-ary predicate $R$ in the
$\SfFOL$ language, the two-sorted language has an $n+1$-ary
predicate~$R$ with $n$ arguments of sort $\domain$ and one additional
argument of sort $\world$. This language $\mathcal{SL}$ is interpreted
in the standard way over two-sorted first-order structures; for the
two-sorted language induced by the correspondence language of
$\SfALCloc$, these are just $S5$-interpretations. One has a
translation $(-)^{\dagger_v}$ of $\SfFOL$ into the two-sorted
first-order language, given by
$R(x_1,\dots,x_n)^{\dagger_v}=R(x_1,\dots,x_n,v)$ and
$(\Box\phi)^{\dagger_v}=\forall v.\,(\phi^{\dagger_v})$, and
commutation with all other constructs, where $v$ is a variable of sort
$\world$. Sturm and Wolter~\shortcite{SturmWolter01} show that
$\SfFOL$ is, over unrestricted $S5$-interpretations, precisely the
fragment of $\mathcal{SL}$ that is determined by invariance under
\emph{potential $S5$-isomorphisms}, i.e.\ unbounded
$S5$-Ehrenfeucht-Fra\"iss\'e equivalence, defined as above but without
a bound on the number of rounds. Since every potential
$S5$-isomorphism is a bisimulation, we can combine this result with
Theorem~\ref{thm:rosen} to obtain that $\SfALCloc$ is the bisimulation
invariant fragment of~$\mathcal{SL}$:
\begin{cor}[Modal characterization within
  $\mathcal{SL}$]\label{cor:benthem}
  Let $\phi = \phi(x,v)$ be a $\bisim$-invariant formula with one
  variable $x$ of sort $\domain$ and one free variable $v$ of sort
  $\world$, in the two-sorted first-order language
  $\mathcal{SL}$. Then there exists an $\SfALCloc$-concept $C$ such
  that $\phi$ is logically equivalent to $(\ST_x(C))^{\dagger_v}$.
\end{cor}
\noindent (Unlike for Theorem~\ref{thm:rosen}, there is as yet no
version of Corollary~\ref{cor:benthem} for finite
$S5$-interpretations, as the characterization of $\SfFOL$ within
$\mathcal{SL}$ is known only for the unrestricted case.)

\section{Conclusions}
We have proved a modal characterization theorem for the modal
description logic $\SfALCloc$, i.e.\ $\SfALC$~\cite{GabbayEA02} with
only local roles.  Specifically, we have shown that $\SfALCloc$, one
of the modal description logics originally introduced by Wolter and
Zakharyaschev~\shortcite{WolterZakharyaschev99}, is, both over finite
and over unrestricted models, the bisimulation-invariant fragment of
$S5$-modal first-order logic. By a result of Sturm and
Wolter~\shortcite{SturmWolter01}, it follows moreover that $\SfALCloc$
is, over unrestricted models, the bisimulation-invariant fragment of
two-sorted FOL with explicit worlds. To our knowledge, these are the
first modal characterization theorems in modal description logic.

It remains a topic of interest to obtain similar characterization
theorems for other modal description logics or many-dimensional modal
logics. Notably, this concerns logics whose modal dimension differs
from the comparatively simple structure of $S5$, e.g.\ $K_\ALC$. Also,
one may investigate the possibility of a modal characterization of
full $\SfALC$, then of course with respect to a different notion of
equivalence.

\bibliographystyle{mynamed}
\bibliography{coalgml}
\newpage
\appendix

\section{Details and Proofs}
\subsubsection{Details for Example~\ref{expl:modalized-roles}}

We show that the $\SfALC$-concept $\exists\Diamond r.\, A$ fails to be
invariant under bisimulation. Define an $S5$-interpretation $\CI$
by taking $\Delta^\CI=\{a,b\}$, $W^\CI=\{v_1,v_2\}$,
$r^{\CI,v_2}=\{(a,b)\}$, $r^{\CI,v_1}=\emptyset$, $A^{\CI,v_1}=\{b\}$,
and $A^{\CI,v_2}=\emptyset$. Moreover, define an
$S5$-interpretation $\CJ$ by $\Delta^\CJ=\{a,b\}$,
$W^\CJ=\{w_1,w_2,w_3\}$, $r^{\CJ,w_2}=\{(a,b)\}$,
$r^{\CJ,w_1}=r^{\CJ,w_3}=\emptyset$, $A^{\CJ,w_3}=\{b\}$, and
$A^{\CI,w_1}=A^{\CI,w_2}=\emptyset$. Then we have
$(v_1,a)\bisim(w_1,a)$, as 
\begin{align*}
  R = \{ & ((v_1,a),(w_1,a)),\\
  &((v_1,a),(w_3,a)),\\
  &((v_2,a),(w_2,a)),\\
  &((v_1,b),(w_3,b)),\\
  &((v_2,b),(w_1,b)),\\
  &((v_2,b),(w_2,b))\}
\end{align*}
is a bisimulation. But $a\in(\exists \Diamond r.\,A)^{\CI,v_1}$ while
$a\notin(\exists \Diamond r.\,A)^{\CJ,w_1}$.

\subsubsection{Proof of Lemma~\ref{lem:abisim}}

For item 1 and the `only if' direction of item 3, we note that in the
alternating game only the options for $S$ are restricted when
compared to the ordinary game, in the sense that he is forced to make
$\Delta$-moves at certain times. Also, the total number of pairs of
moves in the $n$-round alternating game is at most $2n+1$. Therefore
$D$ can use his winning strategy for the $\bisim_{2n+1}$ game to win
the $\abisim_n$ game.

For item 2 and the `if' part of item 3, if $D$ has a winning strategy
in the alternating game, then for every winning configuration
$((w,d),(v,e))$ that can occur before the first phase of a round there
must exist functions $f: W^\CI \rightarrow W^\CJ$ and
$g: W^\CJ \rightarrow W^\CI$ such that for any $W$-move to some
$w' \in W^\CI$ the answer by $D$ (according to the strategy) is
$f(w')$ and for any $W$-move to some $v' \in W^\CJ$ the answer by $D$
is $g(v')$. Now, as long as $S$ keeps making $W$-moves, $D$ can just
reply according to the functions $f$ and $g$. If $S$ does so
indefinitely (in the unbounded game), then $D$ wins. Otherwise, eventually
either $S$ makes a $\Delta$-move or the game ends. In the latter case,
$D$ wins immediately. In the former case, there are two subcases:
\begin{myitemize}
\item There were no $W$-moves played. Then $D$ is in the same
situation that arises in the alternating game when $S$ decides against
moving in the first phase of a round. $D$ plays his winning reply for
that situation.
\item There was at least one pair of $W$-moves, w.l.o.g.\ the last one
was $S$ picking $w' \in W^\CI$ and $D$ replying with $f(w')$. This is
the same configuration that arises in the alternating game when $S$
plays  $w'$ during the first phase, by definition of $f$. So $D$ has
a winning reply for $S$'s $\Delta$-move.
\end{myitemize}
For the finite case it should be noted that every configuration that
can be reached in $n$ rounds of the ordinary game according to this
strategy can also be reached in $n$ rounds of the alternating game
(while following the winning strategy). \hfill\qed

\subsubsection{Proof of Lemma~\ref{lem:efequiv}}

The proof will proceed by induction on the structure of $\phi$.
However, we first need to generalize some notions for the purpose of
the proof:

First, we generalize the Ehrenfeucht-Fra\"iss\'e game to allow for more
possible starting configurations:
\begin{defn}
  Let $\CI$, $\CJ$ be $S5$-interpretations, let $w \in W^\CI$, $v \in
  W^\CJ$, let $\bar d$ and $\bar e$ be finite sequences of equal
  length over $\Delta^\CI$ and $\Delta^\CJ$, respectively, and let $n
  \ge 0$.  The \emph{$n$-round Ehrenfeucht-Fra\"iss\'e game for
  $\CI,w,\bar d$ and $\CJ,v,\bar e$} is played with the same rules as
  in Definition~\ref{defn:ef}, but the starting configuration is now
  $((w,\bar d),(v,\bar e))$. We also write $\CI,w,\bar d \ef_n
  \CJ,v,\bar e$ when $D$ has a winning strategy for this game.
\end{defn}

We can now also generalize the notion of $\ef_n$-invariance to
formulas that may have more than one free variable:
\begin{defn}
  Let $\phi$ be an $\SfFOL$-formula with free variables contained in
  $\{x_1,\dots,x_k\}$. $\phi$ is \emph{$\ef_n$-invariant}, if for all
  $S5$-interpretations $\CI$, $\CJ$ and all $w \in W^\CI$, $v \in
  W^\CJ$, $\bar d \in (\Delta^\CI)^k$ and $\bar e \in (\Delta^\CJ)^k$
  such that $\CI,w,\bar d \ef_n \CJ,v,\bar e$,
  \begin{equation*}
    \CI,w,\bar d \models \phi \Leftrightarrow \CJ,v,\bar e \models
    \phi.
  \end{equation*}
\end{defn}

\noindent We are now set to prove the following more general version
of Lemma~\ref{lem:efequiv}. We will denote the rank of a formula
$\phi$ by $\rk(\phi)$.

\begin{lem}
Every $\SfFOL$-formula $\phi$ of rank at most $n$ with free variables
contained in $\{x_1,\dots,x_k\}$ is $\ef_n$-invariant.
\end{lem}

\begin{proof}
  The proof will be by induction on the structure of $\phi$.
  \begin{myitemize}
  \item For the base cases where $\phi$ is of the form $y_1=y_2$ or
  $R(y_1,\dots,y_m)$ the $\ef_n$-invariance follows from the fact that
  the starting configuration is a partial isomorphism.
  \item The Boolean cases ($\phi = \psi \sqcap \chi$ or $\phi = \neg
  \psi$) are straightforward.
  \item Suppose $\phi = \square \psi$. Since $\rk(\phi) \le n$, we get
  that $\rk(\psi) \le n-1$ and therefore $\psi$ is
  $\ef_{n-1}$-invariant by the induction hypothesis.
  Let $\CI,w,\bar d \ef_n \CJ,v,\bar e$ and $\CI,w,\bar d \models
  \square \psi$. We then need to show that $\CJ,v,\bar e \models
  \square \psi$, so let $v' \in W^\CJ$ and we need to show that $\CJ,v',\bar e
  \models \psi$. By assumption, $D$ has a winning response if $S$
  plays the $W$-move $v'$, let this response be $w'$. Then
  $\CI,w',\bar d \ef_{n-1} \CJ,v',\bar e$, because $D$ has a winning
  strategy for the remaining $n-1$ rounds. Since $\CI,w,\bar d \models
  \square \psi$, we get $\CI,w',\bar d \models \psi$, and
  $\ef_{n-1}$-invariance of $\psi$ then yields $\CJ,v',\bar e \models
  \psi$, as desired.
  \item Suppose $\phi = \exists x_{k+1}.\psi$ (w.l.o.g. we can
  substitute the variable bound by the quantifier). Since $\rk(\phi)
  \le n$, we get that $\rk(\psi) \le n-1$ and therefore $\psi$ is
  $\ef_{n-1}$-invariant by the induction hypothesis.  Let $\CI,w,\bar
  d \ef_n \CJ,v,\bar e$ and $\CI,w,\bar d \models \exists x_{k+1}. \psi$. We
  then need to show that $\CJ,v,\bar e \models \exists x_{k+1}.\psi$.
  $\CI,w,\bar d \models \exists x_{k+1}.\psi$, so by definition there
  must exist some $d \in \Delta^\CI$ such that $\CI,w,\bar d d \models
  \psi$. By assumption, $D$ has a winning response if $S$ plays the
  $\Delta$-move $d$, let this response be $e$. Then $\CI,w,\bar d d
  \ef_{n-1} \CJ,v,\bar e e$, because $D$ has a winning strategy for the
  remaining $n-1$ rounds. Because of this, it follows that $\CJ,v,\bar
  e e \models \psi$ and thus also $\CJ,v,\bar e \models \phi$, as
  desired.
  \end{myitemize}
\end{proof}

\subsubsection{Proof Details for Lemma~\ref{lem:mc1}}

We first note that $\CI,w,d \bisim \CI',w,d$ and $\CJ,v,e \bisim
\CJ',v,e$, where in both cases the bisimulation is given by the
embedding into the disjoint union (and identity on the set of
worlds). So at the end of the proof we can combine $\bisim$-invariance
of $\phi$ with $\CI',w,d \ef_n \CJ',v,e$ to prove:
\begin{equation*}
\CI,w,d \models \phi \Leftrightarrow
\CI',w,d \models \phi \Leftrightarrow
\CJ',v,e \models \phi \Leftrightarrow
\CJ,v,e \models \phi
\end{equation*}

\noindent Now we recall the invariant that $D$ needs to maintain:
  
\begin{quote}
  If $((w,\bar d),(v,\bar e))$ is the current configuration, and
  $\bar d = (d_0,\dots,d_i), \bar e = (e_0,\dots,e_i)$, then $w=v$
  and there is an isomorphism between $\CI'|_{U^{\ell_i}(\bar d)}$
  and $\CJ'|_{U^{\ell_i}(\bar e)}$ mapping each $d_j$ to $e_j$.
\end{quote}
\noindent First, the invariant clearly holds at the beginning of the
game, as the starting configuration is $((w_0,d_0),(w_0,d_0))$, and
since $\ell_0 = \ell$, both interpretations from the invariant are
isomorphic to $\CJ$ and that isomorphism maps $d_0$ to itself.

Second, whenever the invariant holds after at most $n$ rounds, the
current configuration is a partial isomorphism as defined in
Definition~\ref{defn:ef}, i.e.\ actually ensures that~$D$ wins. Using
the names from the invariant, this follows directly from the fact that
the isomorphism maps every $d_j$ to the corresponding $e_j$, where for
the second item in the definition of partial isomorphism we note that
$w=v$.

Finally, we show that the invariant is actually invariant with respect
to the strategy described in the proof sketch. For the case of a
$W$-move, this is clear. In the following, we treat the case of a
$\Delta$-move, with notation as in the invariant. There are two cases:

First, suppose that $S$ picks $d \in U^{2\ell_{i+1}}(\bar d)$ and $e$
is $D$'s response according to the isomorphism. Note that, using the
triangle inequality for the Gaifman distance,
$U^{\ell_{i+1}}(d) \subseteq U^{\ell_i}(\bar d)$ (since
$2\ell_{i+1}+\ell_{i+1}=3\ell_{i+1}=3\cdot3^{n-i-1}=3^{n-i}=\ell_i$)
and thus also $U^{\ell_{i+1}}(e) \subseteq U^{\ell_i}(\bar e)$ by
isomorphism. This implies that the domain $U^{\ell_{i+1}}(\bar d d)$
and range $U^{\ell_{i+1}}(\bar e e)$ of the putative new isomorphism
are contained in those of the old one.  Therefore, the new isomorphism
can be taken to be the restriction of the old isomorphism to the new
domain and range. The same argument works if $S$ picks some
$e \in U^{2\ell_{i+1}}(\bar e)$ instead.

Otherwise, $S$ picks a $d$ such that $D(d_j,d) > 2\ell_{i+1}$ for all
$0 \le j \le i$. Then,
$U^{\ell_{i+1}}(\bar d) \cap U^{\ell_{i+1}}(d) = \emptyset$, again by
the triangle inequality. Now $D$ picks $e$ in $\CJ'$ from a fresh copy
(which means that it contains none of the $e_j$ ($0 \le j \le i)$) of
the same type ($\CI$ or $\CJ$) that $d$ lies in. Such a copy always
exists, because $\CJ'$ contains $n$ copies of both types and in each
of the $n$ rounds at most one of them is visited. Now we obtain two
isomorphisms of $S5$-interpretations. The
radius-$\ell_{i+1}$-neighbourhoods of $\bar d$ and $\bar e$ are
isomorphic by restriction of the old isomorphism, as in the first
case.  The radius-$\ell_{i+1}$-neighbourhoods of $d$ and $e$ are
isomorphic because $d$ and $e$ are the same element in isomorphic
copies of the same type ($\CI$ or $\CJ$). Now, since the domains and
ranges of the two isomorphisms are disjoint, we can combine them into
a new isomorphism which satisfies the constraints from the
invariant. Again, the same argument applies for the case where $S$
picks an element $e$ in $\CJ'$ instead. \hfill\qed

\subsubsection{Proof Details for Lemma~\ref{lem:mc2}}

We first show the following lemma:

\begin{lem}
  Let $\CI$ be an $S5$-interpretation, $(w,d) \in
  W^\CI\times\Delta^\CI$. Then $\CI,w,d \abisim_\ell
  \CI|_{U^\ell(d)},w,d$.
\end{lem}
\begin{proof}
  The winning strategy for $D$ is to copy every move by $S$. Clearly,
  no winning configuration for $S$ can be reached in this way, so we
  just need to show that this is a valid strategy, i.e. that copying
  $S$'s move is always a legal move. But this easily follows from the
  fact that after $k$ rounds of the game, if the current configuration
  is $((w',d'),(w',d'))$ then $D(d,d') \le k$. Note that any $W$-moves
  $S$ elects to play (including the one after the $\ell$-th round) do
  not affect the Gaifman distance.
\end{proof}

Now let $\CI,w,d$ and $\CJ,v,e$ be as in Lemma~\ref{lem:mc2}, so 
$\CI,w,d \abisim_\ell \CJ,v,e$. Because every $S5$-interpretation is
bisimilar to its tree unravelling, and bisimilarity implies
alternating bisimilarity up to depth $\ell$ by Lemma~\ref{lem:abisim},
$\CI,w,d \abisim_\ell \CI^\ast_d,w,d$ and $\CJ,v,e \abisim_\ell
\CJ^\ast_e,v,e$. By transitivity of $\abisim_\ell$, $\CI^\ast_d,w,d
\abisim_\ell \CJ^\ast_e,v,e$. 

Using the above lemma and again transitivity of $\abisim_\ell$, we
obtain $\CI^\ast_d|_{U^\ell(d)},w,d \abisim_\ell
\CJ^\ast_e|_{U^\ell(e)},v,e$.



Now we show that the winning strategy for $D$ in the $\abisim_\ell$
game between $\CI^\ast_d|_{U^\ell(d)}$ and $\CJ^\ast_e|_{U^\ell(e)}$
is also winning for the $\abisim$ game. The tree structure of the
interpretations guarantees that, regardless of strategy, after round
$k$, if the current state is $((w',d'),(v',e'))$, then $d'$ and $e'$
are at distance $k$ from the root. So, if $D$ follows his winning
strategy, either $S$ loses within $\ell$ rounds or at least $\ell$
rounds are played. Going into round $\ell+1$, there are two cases for
phase 1 of this round:
\begin{myitemize}
\item $S$ chooses to make a $W$-move. In this case, there exists a
winning reply for $D$, remembering that the $\abisim_\ell$ game allows
for a last pair of $W$-moves to be played after round $\ell$, so this
case is covered by the existing strategy.
\item $S$ does not choose to make a $W$-move. In this case, we go
straight to phase 2.
\end{myitemize}
Now, $S$ is forced to make a $\Delta$-move, but both individuals in the
current configuration are at distance $\ell$ from their respective
roots $d$ and $e$, so they do not have any $r$-successors for any role
$r$. Therefore $S$ cannot make a legal move, and $D$ wins the game.

By item 3 of Lemma~\ref{lem:abisim}, $\CI^\ast_d|_{U^\ell(d)},w,d \bisim
\CJ^\ast_e|_{U^\ell(e)},v,e$, so to finish the proof, we combine the
$\bisim$-invariance and $\ell$-locality of
$\phi$ as follows:
\begin{multline*}
\CI,w,d \models \phi \Leftrightarrow
\CI^\ast_d,w,d \models \phi \Leftrightarrow
\CI^\ast_d|_{U^{\ell}(d)},w,d \models \phi \Leftrightarrow \\
\CJ^\ast_e|_{U^{\ell}(e)},v,e \models \phi \Leftrightarrow
\CJ^\ast_e,v,e \models \phi \Leftrightarrow
\CJ,v,e \models \phi
\end{multline*}
\hfill\qed

\subsubsection{Details for Remark~\ref{rem:partunravel}}

Let $\CI$ be an $S5$-interpretation and $(w_0,d_0) \in
W^\CI\times\Delta^\CI$. Let $\CJ,w_0,d_0$ be the partial unravelling of
$\CI$ up to level $\ell$. Then we can define a map $\rho: \Delta^\CJ
\rightarrow \Delta^\CI$ as follows: every element from a copy of
$\CI$ is mapped to itself and every path $\bar d$ from the tree
unravelling is mapped to its last element $\pi(\bar d)$ (note that
this is well-defined, because any elements that were identified to
form the partial unravelling have the same image under this map).

A bisimulation is then given by $\rho$ in the individual dimension and
identity in the world dimension, i.e.
\begin{equation*}
(w,d) R (v,e) \Leftrightarrow w = v \text{ and } \rho(e) = d
\end{equation*}

\subsubsection{Proofs of the Properties of Normal Forms}

We prove the following two properties of normal forms:

\begin{myitemize}
\item For any $\CI,w,d$, there is exactly one normal form
  $C^k_{\CI,w,d}$ of rank $k$ such that $\CI,w,d \models
  \nf{k}{\CI,w,d}$.
\item We have $\CI,w,d \bisim_k \CJ,v,e$ iff
  $\nf{k}{\CI,w,d} = \nf{k}{\CJ,v,e}$.
\end{myitemize}

\noindent In what follows, we will sometimes refer to the
$\varepsilon_B$ from Definition~\ref{defn:nf} as \emph{signs} where
`nothing' is the positive sign and negation the negative sign.

For the first property: For every $B \in \mathsf{at}_k$, either
$\CI,w,d \models B$ or $\CI,w,d \models \neg B$, and we put
$\varepsilon_B$ to be nothing in the first case and negation in the
second. Together, this gives a normal form $C^k_{\CI,w,d}$. For
uniqueness, we note that if we defined any of the $\varepsilon_B$
differently, $\CI,w,d$ would fail to satisfy the resulting normal
form.

For the second property: For the `only if' direction, we note that if
$\CI,w,d \bisim_k \CJ,v,e$, then $\CJ,v,e \models C^k_{\CI,w,d}$ by
$k$-bisimilarity (and the fact that $C^k_{\CI,w,d}$ is of rank $k$),
but then $C^k_{\CI,w,d} = C^k_{\CJ,v,e}$ by uniqueness of normal forms
(the first property).  For the `if' direction, we proceed by
induction on $k$. So suppose $C^k_{\CI,w,d} = C^k_{\CJ,v,e} =: C$.
First we ensure that the configuration $((w,d),(v,e))$ is not winning
for $S$: for every atomic concept $A$, the sign of $A$ in $C$
determines for both sides whether they satisfy $A$ or not, so
$\CI,w,d$ and $\CJ,v,e$ satisfy the same atomic concepts.  If $k=0$,
we are done because the game ends immediately. So suppose $k>0$, and
we now need to give a winning response for $D$ for all possible moves
by $S$:
\begin{myenumerate}
\item Suppose $S$ picks some $w' \in W^\CI$. Then $\CI,w,d \models
\Diamond C^{k-1}_{\CI,w',d} =: B$, so the sign $\varepsilon_B$ in $C$
is positive and so also $\CJ,v,e \models
\Diamond C^{k-1}_{\CI,w',d}$. By definition, there must exist
some $v' \in W^\CJ$ such that $\CJ,v',e \models C^{k-1}_{\CI,w',d}$.
By uniqueness of normal forms, $C^{k-1}_{\CI,w',d} =
C^{k-1}_{\CJ,v',e}$, and by the induction hypothesis, $\CI,w',d
\bisim_{k-1} \CJ,v',e$. This means that $D$ has a winning strategy for
the remaining $k-1$ rounds, and thus wins the bisimulation game.

\item The case where $S$ picks some $v' \in W^\CJ$ is
  analogous. 

\item Suppose $S$ picks the role $r\in\RN$ and $d'\in\Delta^\CI$
such that $(d,d')\in r^{\CI,w}$. Then $\CI,w,d \models
\exists r.C^{k-1}_{\CI,w,d'} =: B$, so the sign $\varepsilon_B$ in $C$
is positive and so also $\CJ,v,e \models
\exists r.C^{k-1}_{\CI,w,d'}$. By definition, there must exist
some $e' \in \Delta^\CJ$ such that $\CJ,v,e' \models C^{k-1}_{\CI,w,d'}$.
By uniqueness of normal forms, $C^{k-1}_{\CI,w,d'} =
C^{k-1}_{\CJ,v,e'}$, and by the induction hypothesis, $\CI,w,d'
\bisim_{k-1} \CJ,v,e'$. This means that $D$ has a winning strategy for
the remaining $k-1$ rounds, and thus wins the bisimulation game.

\item The case where $S$ picks a role $r\in\RN$ and $e'\in\Delta^\CJ$
  such that $(e,e')\in r^{\CJ,v}$ is
  analogous. 
  \hfill\qed
\end{myenumerate}
\end{document}

%% file: defsplain.tex
\theoremstyle{definition}
       \newtheorem{defn}{Definition}[section]
       \newtheorem{expl}[defn]{Example}

       \newtheorem{rem}[defn]{Remark}
\theoremstyle{plain}
       \newtheorem{thm}[defn]{Theorem}
       
       \newtheorem{lem}[defn]{Lemma}

       \newtheorem{cor}[defn]{Corollary}